\documentclass[11pt]{article}
\usepackage{amsmath,amsthm,amsfonts}
\usepackage{color}

\title{An Online Multi-unit Auction with Improved Competitive Ratio}

\author{Sourav Chakraborty\thanks{Technion, Israel. Email:
\hbox{sourav@cs.technion.ac.il}} \and Nikhil R
Devanur\thanks{Microsoft Research, Redmond, USA. Email:
\hbox{ndevanur@gmail.com}
\newline Initial part of the work was done when both the authors
were visiting Microsoft Research, India in Summer 2007}}

\date{}

\newtheorem{theorem}{Theorem}[section]
\newtheorem{lemma}[theorem]{Lemma}

\newtheorem{corollary}[theorem]{Corollary}

\newtheorem{definition}[theorem]{Definition}

\def\N{{\mathbb N}}

\def\R{{\mathbb R}}

\def\wait{{\sc WAIT}}
\def\allocate{{\sc ALLOCATE}}

\newcommand{\ignore}[1]{}
\newcommand{\beq}{\begin{equation}}
\newcommand{\eeq}{\end{equation}}

\begin{document}
\maketitle

\begin{abstract}
We improve the best known competitive ratio (from 1/4 to 1/2),
for the online multi-unit allocation problem,
where the objective is to maximize the single-price revenue.
Moreover, the competitive ratio of our algorithm tends to 1,
as the bid-profile tends to ``smoothen''.
This algorithm is used as a subroutine in designing truthful auctions
for the same setting: the allocation has to be done online,
while the payments can be decided at the end of the day.
Earlier, a reduction from  the auction design problem to the allocation problem
was known only for the unit-demand case.
We give a reduction for the general case when the bidders have
decreasing marginal utilities.
The problem is inspired by sponsored search auctions.

\end{abstract}

\section{Introduction}
We improve the best known competitive ratio (from 1/4 to 1/2),
for the online multi-unit allocation problem, which
in turn gives a factor of 2 improvement for the online
multi-unit auction problem.
Moreover, the competitive ratio of our algorithm tends to 1,
as the bid-profile tends to ``smoothen''.
We also give a reduction from the auction problem to the
allocation problem when the bidders want multiple copies
of the item with decreasing marginal utilities for them.
Earlier, such a reduction was known only for the unit-demand case.

\begin{definition}{\bf Online Multi-unit Auction Problem for single
demand} We have for auction multiple copies of a single item, where
the copies are coming online. We have no prior knowledge of how many
copies of the item will be produced. Each bidder has a utility $u_i$
for one copy of the good, and bids a value $b_i$ at the beginning of
the auction.

The problem is to design an allocation and pricing scheme that has
the bids $b_i$ as input and has the following properties:
\begin{itemize}\item (Truthfulness) The auction mechanism should be
truthful.
\item (Perishable Good) As a new copy comes, either it is
allocated to some bidder, who wins the copy, or it is discarded.
\item The prices charged to the winning  bidders are determined at the
end of the auction, when there are no more copies left.
\end{itemize}

The goal of the auction is to maximize the revenue of the
auctioneer.
\end{definition}

The corresponding online multi-unit allocation Problem for single
demand is the following:

\begin{definition}{\bf Online Multi-unit Allocation Problem for single
demand} We want to sell multiple copies of a single item, where the
copies are coming online. We have no prior knowledge of how many
copies of the item will be produced. Each interested bidder sends
her bid $b_i$ at the beginning of the auction.

The problem is to design an allocation scheme that has the bids
$b_i$ as input and as a new copy comes, either it is allocated to
some bidder, who wins the copy, or it is discarded. All the winning
bidders are charged the same amount at the end of the auction and is
less than the bid of all the winning bidders.

The goal of the auction is to maximize the revenue of the
auctioneer.
\end{definition}

A generalization of the above problem is the multiple demand
problem.

\begin{definition}{\bf Online Multi-unit Auction Problem for
Multiple demand} Just like in the single demand problem we have for
auction multiple copies of a single item, where the copies are
coming online. But in this problem the bidders may bid for multiple
copies of the item. The bidders make multiple bids for each copy.

The problem is again to design an allocation and pricing scheme that
satisfies the all the conditions of truthfulness, perishableness of
the item and charging the price at the end. The goal is to maximize
the revenue of the auctioneer.
\end{definition}

\begin{definition} {\bf Competitive Ratio} Given the bids $\{b_i\}$ and
the number of copies of the item that is produced let $OPT$ be the
revenue of the optimum single price auction. If the expected revenue
generated by an online multi-unit auction mechanism is $REV$ then
the {\em Competitive Ratio} of the mechanism is $REV/OPT$.
\end{definition}

\begin{theorem} We design an allocation algorithm for the online
multi-unit auction problem for unit demand, that achieves a
competitive ration of $1/2$.
\end{theorem}

\begin{corollary} We have a truthful auction mechanism for the
online multi-unit auction problem, that achieves a constant
competitive ratio.
\end{corollary}

In the {\em Online Multi-unit Auction} problem,
we have for auction multiple copies of a single item,
where the copies are coming online.
We have no prior knowledge of how many copies of the
item will be produced.
Each bidder has a utility $u_i$ for one copy of the good,
and bids a value $b_i$ at the beginning of the auction.
As a new copy comes, either
it is allocated to some bidder, who wins the copy,
or it is discarded.
The prices charged to the winning  bidders are determined
at the end of the auction, when there are no more copies left.
The goal of the auction is to maximize the revenue of the auctioneer.
We use competitive analysis, and compare this revenue
to the optimal ``single price'' revenue, that is the optimal
revenue that could have been obtained by charging the
same  price
for all the copies allocated,
having known the number of copies produced and
the true utilities of the bidders.
Further, we want the auction to be strategyproof, that is,
the auction is so designed that bidding truthfully (that is $b_i =u_i$)
is a dominant strategy for every bidder.

Mahdian and Saberi \cite{MS} showed that this problem can be reduced
to the {\em Online Multi-unit Allocation} problem, with a constant
factor lost in the competitive ratio. The allocation problem is
exactly as before, except that one can assume that the bidders bid
their true utility, and the allocation algorithm is also restricted
to be single priced, the price being determined at the end of the
algorithm. As before, the revenue of the algorithm is compared to
the optimum single price revenue, having known the number of copies.

If a bidder wants more than one copy of the item, then he
submits multiple bids. The allocation problem is essentially the
same. However, the auction problem is more difficult because
of truthfulness. The auction has to now consider the possibility
that the bidder lies about some subset of his bids, as opposed
to lying about the only bid in the unit-demand case.

The problems are inspired by sponsored search auctions, which are a
major source of revenue for search engines like Google, Yahoo and
MSN. The copies correspond to search queries and the bidders are the
advertisers. The auction problem considered here is also a natural
extension of the line of work on digital goods auction: from
unlimited supply (\cite{GH2, GH3, GHKSW, GHW}) to limited supply
(\cite{abrams, BNDHT, borgs}), to unknown supply (\cite{MS} and this
paper). For the sponsored search auction setting, the multiple
demand case is more realistic and our reduction for this case is of
significant interest.

It can be easily seen that the competitive ratio of any deterministic
algorithm for the allocation problem is arbitrarily small.
So it is actually surprising that
a randomized algorithm can even get a constant competitive ratio.
The reason for this difficulty is that the revenue of the algorithm,
as a function of the number of copies allocated can have many ``peaks''
and ``valleys''.  For any deterministic algorithm, an adversary can
make sure that the algorithm either ends up in a valley, or is stuck
on a small peak while the optimum is at a larger peak elsewhere.
Our allocation algorithm actually solves the following general online problem:
let $f:\N \rightarrow \R$ be any ``sub-linear''\footnote{
what we mean is $f(n)/n$ is decreasing. A function of the form
$f(n) = cn^a$ satisfies this condition if and only if $a$ is at most $1$.}
function of natural numbers.
Consider an online algorithm that is trying to choose $n$ to maximize $f(n)$.
The algorithm starts at $n=0$, and at each time step, can choose to
stay at $n$, or go to $n+1$. The competitive ratio of the algorithm is
\[ \frac{f(n)}{max_{1\leq i \leq m} \{f(i)\}} \]
where $m$ is the total number of time steps.
We give an algorithm that achieves
a competitive ratio of $1/2$ for this problem.

The simplicity of our algorithm is quite appealing.
Whenever the algorithm is at a peak, it has to decide if it
has to stay at the peak, or try to get to the next one.
What the algorithm does is to simply wait at the current peak
for a period of time chosen uniformly at random between
1 and the maximum distance between peaks seen so far.

The proof of the competitive ratio relies on case analysis
since the optimal revenue and the expected revenue of the algorithm
vary depending on the total number of copies seen.
A good idea of how the analysis goes can be had by considering the following
instance: suppose there is one bid of $1$ and many bids of $\epsilon \ll 1$.
In this case the algorithm waits for a time chosen u.a.r between
1 and $1/\epsilon$.
If the number of copies seen is $m \leq 1/\epsilon$,
then the optimal revenue is 1, while the
expected revenue is $1-x +\frac{x^2}{2}$ (where $x = \epsilon m$),
which is at least $1/2$ when $x \leq 1$.
If $m \geq 1/\epsilon$ then the optimal revenue is $\epsilon m$,
while the expected revenue is $\epsilon m - 1/2 \geq \frac{\epsilon m}{2}.$

The reduction in \cite{MS} gives a
constant competitive ratio for the auction problem. The auction is
based on random sampling with computing optimal  ``price offers''.
But when run in an online setting, the prices offered decrease over
time, due to which a bidder might regret not getting a copy earlier
as the price decreased at a later time. The authors of \cite{MS}
take care of this situation by a clever implementation that works
only when all bidders want only one copy. It is not truthful when
the bidders can submit multiple bids. We circumvent this difficulty
by combining the random sampling technique with the VCG auction.
However, we only get an asymptotic competitive ratio, that is the
ratio tends to 1, as a certain bidder dominance parameter tends to
0.



The problems we consider were first
studied by Mahdian and Saberi \cite{MS}. Other variants have been
considered, for instance when the supply is given while the bidders
arrive and leave  online \cite{HKP, blum}. Also there is a huge work
on digital goods and random sampling auctions. Another interesting
case is when the bidders have  constant marginal utilities for the
copies, but have daily budgets. \cite{borgs, abrams} gave an auction
for this case with known supply.
Extending it to the online setting is an important open problem.

\section{An Algorithm for the Online Multi-unit Allocation Problem}
Without loss of generality  assuming that the bids are
$u_1 \geq u_2 \geq \dots \geq u_n$, the revenue obtained by
allocating $l$ units of the item is $lu_l$.
Let $1=a_1 < b_1 <
a_2 < b_2 < a_3 < b_3 < \dots$  be the critical points of the
function $lu_l$, that is, the function $l u_l$ is
non-decreasing as $l$ increases from
$a_i$ to $b_i$, for all $b_i < l  < a_{i+1}$ we have $b_i
u_{b_i} > l u_{l}$ and $b_i u_{b_i} \leq a_{i+1} u_{a_{i+1}} .$


The algorithm is in one of two states, \allocate\ or \wait. When it
is in \allocate, it allocates the next copy of the item. When it is
in \wait, it discards the next copy.
The description of the algorithm is completed by specifying when it
transits from one state to the other.

The algorithm is initially in  \allocate. It transits from
\allocate\  to  \wait\ when the number of copies allocated ($X$) is
equal to $b_i$ for some $i$. It transits from \wait\ to \allocate\
when the number of copies discarded ($Y$) is equal to a random
variable, $T$, for waiting time.
$T$ is reset every time the algorithm transits to \wait.
$T$ is picked so that it is distributed uniformly between 0 and
$D_i$, where
$$D_i= \max_{j\leq i}(a_{j+1} - b_j)$$
(recall that $X =b_i$).
We further want to maintain the invariant that $Y$ never exceeds $T$. Equivalently, the value of $T$ can only
increase during a run of the algorithm.

We still have to specify how $T$ is picked. Because of the condition
that $T$ can only increase, we cannot pick $T$ independently every
time we transit to \wait.
Note that if $D_i = D_{i-1}$, then we don't have to change $T$ at
all.
If $D_i > D_{i-1}$, then  \begin{itemize}
\item w.p.  $\frac{D_{i-1}}{D_i}$ don't change $T$,
\item with the remaining probability pick $T$ uniformly at random
from the interval $[D_{i-1},D_i]$.
\end{itemize}
It is easy to see that the resulting $T$ is distributed uniformly in
$[0,D_i$]. Note that in case $T$ is not changed, then $Y$ is already
equal to $T$, and we transit back to \allocate\ immediately.
Equivalently, we don't transit to \wait\ at all.

\subsection*{Pseudocode for the Algorithm}
\noindent 1. initialize STATE = ALLOCATE, i=1, X=Y=T=0;\\
2. when a new copy is produced{\\
3. \hspace{0.25 in} If (STATE = ALLOCATE)\\
4. \hspace{0.5 in} Allocate the copy to the next bidder;\\
5. \hspace{0.5 in} X ++;\\
6. \hspace{0.5 in} If (X = $b_{i}$)\\
7. \hspace{.75 in} If ($D_i > D_{i-1}$) \\
8. \hspace{1 in}   With prob   $1- \frac{D_{i-1}}{D_i}$\\
9. \hspace{1.25 in}  set $T$ to a random
number from the interval $[D_{i-1},D_i]$;\\
10. \hspace{1.25 in} STATE = WAIT;\\
11. \hspace{.75 in} i ++; \\
12. \hspace{0.25 in} If (STATE = WAIT)\\
13. \hspace{.5 in} Discard the copy;\\
14. \hspace{.5 in} Y ++;\\
15. \hspace{.5 in} If (Y $= T$) \\
16. \hspace{.75 in} STATE = ALLOCATE \\
17. \hspace{0.25 in} GO TO line 2.

\section{Competitive Analysis}
In this section we show that the expected revenue of our algorithm, $ALG$,
is at least half of the optimal revenue on hindsight, $OPT$.
Let $M$ be the number of copies that is produced at the end of the
day. Let $b_i < M \leq b_{i+1}$.
\begin{description}
\item [Case 1:] If $M\leq a_{i+1}$, then  $OPT$ is $b_iu_{b_i}$.
\item[Case 2:] If $a_{i+1} < M$,  then $OPT$ is $Mu_M$.
\end{description}
Recall that $X$ is the number of items sold by the algorithm.
Therefore, $ALG = E[Xu_X]$.
We approximate $ALG$ in the two cases as follows.

\noindent In Case 1, if $X\leq b_i$ then $u_X \geq u_{b_{i}}$, and
if $X>b_i$ then $u_X \geq u_{a_{i+1}}$.
Hence we have that
\begin{equation*}
ALG \geq \Pr[X<b_i] E[X|X<b_i] u_{b_i} + \Pr[X=b_i] b_i u_{b_i} + \Pr[X>
b_i]E[X|X>b_i]u_{a_{i+1}}
\end{equation*}
Also note that by definition $a_{i+1}u_{a_{i+1}} \geq b_iu_{b_i} = OPT$.
Hence
\begin{equation*}
u_{a_{i+1}} \geq \frac{b_iu_{b_i}}{a_{i+1}} = \frac{OPT}{a_{i+1}}.
\end{equation*}
Substituting for the values of $u_{b_i}$ and $u_{a_{i+1}}$, we get
\begin{equation}\label{eq:exprev2}
\frac{ALG}{OPT}\geq  \Pr[X<b_i]E[X|X<b_i]\frac{1}{b_i} + \Pr[X=b_i] + \Pr[X>
b_i]E[X|X>b_i]\frac{1}{a_{i+1}}.\end{equation}

\noindent In Case 2,  we use the fact that $u_X \geq u_M$, so
$ALG \geq E[X] u_M$, and since  $OPT = Mu_M$,
we need to prove that  $E[X]$  is  at least $M/2$.
\begin{equation}\label{eq:exprev3}
E[X] = \Pr[X<b_i]E[X|X<b_i] + \Pr[X=b_i]b_i  + \Pr[X>
b_i]E[X|X>b_i].\end{equation}
We now give a way to calculate the various probabilities and
expectations needed.
\begin{definition} For all $i\geq 1$ let $T_i$ be the value of
the random variable $T$ chosen at phase $i$.
\end{definition}
$T_i$ is the number of number of items we plan to discard
before allocating the $(b_i +1)$-th element. Also note that $T_i$ is
distributed uniformly between $0$ and $D_{i}$,
and  $T_{i-1} \leq T_{i}$ for all $i$.
Also, it is easy to see from the description of the algorithm that
$X \geq  M - T_{i-1} $ when $X<  b_i $, and
$X = M - T_i$ when $X> b_i$.
This gives us the following lemmas.
Let $M' := M - b_i$.
\begin{lemma}\label{lem:equiv}
For all $i \geq 1$ the following statements hold,
(let $ T_0 = 0$),
\begin{enumerate}
\item $X < b_i \Leftrightarrow T_{i-1} > M'$.
\item $X = b_i$ $\Leftrightarrow$ $T_{i-1} \leq M'$ and $T_i \geq M'$.
\item $X > b_i$ $\Leftrightarrow$ $(T_{i-1} < M' \mbox{ and }T_i < M')$
$\Leftrightarrow T_i < M'$.
\end{enumerate}
\end{lemma}

\begin{lemma} \label{lem:expect}
\begin{eqnarray*}&&E[X | X < b_i] \geq M - E[T_{i-1}| T_{i-1} > M']
= M - \frac{ M' + D_{i-1}}{2}.\\
&&E[X | X > b_i] = M- E[T_i|T_i < M']
= M - \frac{\min\{M',D_{i}\}}{2}.\end{eqnarray*}
\end{lemma}

However, the probability of the events $X < b_i$, $X=b_i$ and $X>b_i$
depend upon the order of $J_i := b_i +D_{i-1}$,  $a_{i+1}$ and $M$.
So we consider all possible orders of these 3 quantities separately.
Table 1 shows the probabilities for all the cases.
\begin{table}[h]\label{Table1}
\begin{center}
\begin{tabular}{||l | c | c | c ||}
    \hline \hline
    &        &       &\\
    & \( \begin{array}{c} \Pr[X < b_i] \\
      =\Pr[T_{i-1} > M']\end{array}\) & $\Pr[X = b_i]$ &
\(\begin{array}{c} \Pr[X > b_i] \\ =\Pr[T_i < M']\end{array}\) \\
    & &           &  \\ \hline
    &        &       &\\
    Case 1a \( \begin{array} {c}J_i<M\leq a_{i+1}\\
D_{i-1} \leq M' \leq D_i \end{array} \) &    $0$   &
$ 1 - \frac{M'}{D_i}$
    &$\frac{M'}{D_i}$ \\
    &        &       & \\ \hline
    &        &       & \\
    Case 1b \( \begin{array}{c} M\leq J_i< a_{i+1}\\
              M' \leq D_{i-1} \leq D_i \end{array} \) &
$1-\frac{M'}{D_{i-1}}$
    &  $\frac{M'}{D_{i-1}}-\frac{M'}{D_{i}}$ &
  $ \frac{M'}{D_{i}} $  \\
    &        &       &\\ \hline
    &        &       &\\
    Case 1c  \( \begin{array}{c} M\leq a_{i+1}\leq J_i\\
      M' \leq D_{i-1} =  D_i \end{array} \)  &
$1-\frac{M'}{D_{i-1}}$
    & $0$    &  $\frac{M'}{D_{i-1}}$ \\
    &        &       &\\ \hline \hline
    &        &       &\\
    Case 2a  \( \begin{array}{c} J_i<a_{i+1}\leq M\\
 D_{i-1} <  D_i \leq M' \end{array} \) &    0  &   0   &   1    \\
    &        &       &\\ \hline
    &        &       &\\
    Case 2b  \( \begin{array}{c}a_{i+1}\leq M\leq J_i\\
 M' \leq D_{i-1} =  D_i \end{array} \)  &
$1-\frac{M'}{D_{i-1}}$
    & $0$    &  $\frac{M'}{D_{i-1}}$ \\
    &        &       &\\ \hline
    &        &       &\\
    Case 2c \( \begin{array}{c}a_{i+1}\leq J_i<M\\
      D_{i-1} = D_i \leq M' \end{array} \) &  0  &  0  &   1  \\
    &        &       &\\
    \hline
    \hline

    \end{tabular}
\end{center}
\caption{Probability of the event $X< b_i$, $X=b_i$ and $X>b_i$ for the six different cases}
\end{table}

\subsection{Analyzing all the Cases}
A few observations first: from Lemma \ref{lem:expect},
$$E[X | X < b_i] = M - \frac{ M' + D_{i-1}}{2} \geq \frac{M}{2}$$
since $ D_{i-1} \leq b_i = M - M'$ implies $ M' + D_{i-1} \leq M$.
And when $M' \leq D_i$,
$$E[X | X > b_i]
= M - \frac{M'}{2} =  \frac{M+b_i}{2}.$$
\noindent\textbf{Case 1a:} $[J_{i} \leq M \leq a_{i+1}]$
Set $x = \frac{M'}{D_i}$ and $y = \frac{D_i}{a_{i+1}}$.  Then $ \frac{M'}{a_{i+1}} = xy$
and $\frac{b_i}{a_{i+1}} = 1-y$.
Substituting for the probabilities and expectations in (\ref{eq:exprev2}),
$$ \frac{ALG}{OPT} \geq 1-x  + x\left(1+\frac{xy}{2} -y\right)
= 1 + \frac{x^2y}{2} - xy =: \alpha. $$
$$ \frac{d\alpha}{dx} = y(x-1) \leq 0.$$
Therefore $\alpha$ is minimized when $x =1$, and at this point,
$\alpha = 1-y/2 \geq 1/2$ since $y\leq 1$.

\noindent\textbf{Case 1b:} $[b_i \leq M \leq J_i \leq a_{i+1}]$
 As observed earlier, we have that $E[X | X < b_i] \geq \frac{ M}{2} \geq  \frac{b_i}{2} $ and
$E[X | X > b_i] = \frac{M+b_i}{2} \geq b_i. $
Setting $x= \frac{M'}{D_{i-1}}$ and using (\ref{eq:exprev2}) again,
$$\frac{ALG}{OPT} \geq
\left(1-x\right)\frac{1}{2} +
x \left( 1-\frac{D_{i-1}}{D_{i}} +
\frac{D_{i-1}}{D_{i}}  \frac{b_i }{a_{i+1}}\right).$$
It is enough to prove that $\frac{D_{i-1}}{D_{i}} \left(1 -
 \frac{b_i }{a_{i+1}} \right) \leq \frac{1}{2} $.
This follows from the fact that $a_{i+1} = b_i + D_i \geq 2 D_{i-1}$.

\noindent\textbf{Case 1c:} $[b_i \leq M \leq a_{i+1} \leq J_i]$
As before we have that $E[X | X < b_i] \geq \frac{ M}{2} \geq  \frac{b_i}{2} $ and
$E[X | X > b_i] = \frac{M+b_i}{2} \geq \frac{a_{i+1}}{2}. $
The last inequality follows because $a_{i+1} \leq J_i = b_i + D_{i-1}
\leq b_i + M$.
Plugging these back in (\ref{eq:exprev2}) gives
$ \frac{ALG}{OPT} \geq \frac{1}{2}.$

\noindent\textbf{Case 2a: $[J_i \leq a_{i+1}\leq M]$} From (\ref{eq:exprev3}) and Lemma \ref{lem:expect},
$E[X] = M - \frac{D_{i}}{2} \geq  M/2$ since $M \geq D_i$.

\noindent\textbf{Case 2b: $[a_{i+1} \leq M \leq J_i]$}
In this case, it is enough to show that both $E[X | X < b_i]$
and $E[X | X > b_i]$ are bigger than $M/2$. From Lemma \ref{lem:expect},
$E[X| X < b_i ] \geq \frac{M}{2}$.
$E[X| X > b_i ]= M -  \frac{ M'}{2} \geq \frac{M}{2}.$

\noindent\textbf{Case 2c: $[a_{i+1} \leq J_i \leq M]$} The analysis is identical
to Case 2a.

In fact, the competitive ratio of our algorithm is $1-\epsilon$ if
$\epsilon \geq \max \{\frac{D_i-1}{b_i}, \frac{D_i}{a_{i+1}} \}$.
The proof is essentially the same as above.

\section{Designing Truthful Mechanism}
Let ${\bf B} = \{1, 2, \dots, n\}$ be the set of bidders. Each
bidder can make multiple bids. We will design a truthful mechanism
which has good competitive ratio. Our mechanism will use an online
multi-unit allocation algorithm as a sub-routine. Under a
bidder-dominance assumption, the competitive
ratio of our mechanism will be $(1-\epsilon)\alpha$
where $\alpha$ is  the competitive ratio of the
allocation algorithm we use as our subroutine.

\begin{center}
\framebox{
\parbox[s]{6.0in}{
{\bf The Mechanism: }  We divide the set of bidders into two groups
$S$ and $T$ by placing each bidder randomly into either of the
groups. On each set of bidders $S$ and $T$ we will have fictitious
runs of the allocation algorithm. Let the fictitious run of the
allocation algorithm on the set $S$ (respectively $T$) allocates
$x(S,k)$ (respectively $x(T,k)$) copies when $k$ copies
are produced.\\

Now when the $j$-th copy is produced, if $j$ is even we compute
$x(S,j/2)$. If at that time the number of copies allocated to
bidders in $T$ is less than $x(S,j/2)(1-6\gamma)$ then we allocate
the $j$-th copy to $T$ otherwise discard the copy. Similarly, if $j$
is odd we compute $x(T,(j+1)/2)$ and if the number of copies
allocated to bidders in $S$ is less than $x(T,(j+1)/2)(1-6\gamma)$
then we allocate the $j$-th copy to $S$ otherwise discard the
copy.\\

Finally let $x_{final}(S)$ and $x_{final}(T)$ copies are allocated
to bidders in $S$ and $T$ respectively.  The prices charged are the
VCG payments, that is, as if we ran a VCG auction to sell
$x_{final}(S)$ copies to bidders in $S$.  }}
\end{center}

Note that the even indexed copies will be allocated only to bidders
in $T$ and the odd-indexed copies will be allocated only to bidders
in $S$. But the bids of bidders in $S$ decides how many
(odd-indexed) copies will be allocated to bidders in $T$ and vice
versa. This mechanism is similar to that in \cite{GHKSW} on digital
good auction with unlimited supplies except that in \cite{GHKSW} the
bids of bidders in $S$ decides the cut off price for bidders in $T$
and vice-versa.

If $M$ is the number of copies of the item that are finally produced
we denote by $OPT = OPT({\bf B}, M)$ the revenue obtained by the
optimal single price allocation algorithm.

\begin{definition} For any price $p$ and any bidder $i$ we denote by
$n(i,p)$ the number of bids of bidder $i$ that are more than $p$.

We define the bidder dominance parameter $\eta$ as
$${\eta} = \frac{\max_{i,p}n(i,p)p}{OPT}.$$
\end{definition}

\begin{theorem}\label{thm:truthful} The above mechanism is a truthful
mechanism. If all the bids are from a finite set of prices (say Q)
and if $$\frac{1}{\eta} = \Omega \left (
\log\left(\frac{|Q|}{\delta}\right)\left(\frac{1}{\epsilon^2}\right)
\right) $$
and $\gamma = \epsilon/8$ then with probability
more than $(1-\delta)$ our mechanism guarantees a revenue of at
least $\alpha OPT(1-\epsilon)$ on expectation, where $\alpha$ is the
competitive ratio of the allocation algorithm that we use as the
subroutine.
\end{theorem}

In the rest of this section we will give a sketch of
the proof of the theorem. The detailed proof of the theorem is in
the Appendix. The proof is similar to that in \cite{GHKSW}.

The proof that the mechanism is truthful follows from the facts that
the number of copies allocated to each half is independent of the
number of the bids of the bidders in that half and the fact that
pricing is determined by the VCG auction.

The proof of the competitive ratio is in two stages.
The first thing to notice is that since the bidders are split
randomly into two sets the optimal revenue we can obtain from either
of the sets is on expectation nearly half of what we can obtain from
the whole set.

The second thing is that the discounting factor of $(1-6\gamma)$
ensures that w.h.p the eventual winners in $S$ (respectively $T$)
are charged at least as much as our allocation algorithm charges
during its fictitious run on the set $T$ (respectively $S$) .





Note that the bound on the bidder dominance gives us an upper bound
on $n(i,p)$ that is the number of bids on any bidders that is more
than $p$. This is essential for our analysis.

Since the bidders are split randomly into two sets, the optimal
revenue we can obtain from either of the sets is on expectation
nearly half that we can obtain from the whole set. Let $OPT(S,j)$
denote the revenue generated after $j$ copies are produced by a
fictitious run of the optimal single price allocation algorithm on
$S$.  By Mcdiarmid's Inequality and the bound on the bidder
dominance, with probability at least $(1-O(\delta))$ we have $
OPT(S,\lceil M/2 \rceil) > (1/2-\gamma)OPT$, where $M$ is the final
number of copies produced. Similarly we have $ OPT(T,\lfloor M/2
\rfloor) > (1/2-\gamma)OPT$.

For the second stage we again notice that since the set of bidders
was partitioned randomly the number of bids more than $p$ is w.h.p
divided evenly among the two sets $S$ and $T$. Again from the
McDiarmid's Inequality and from the bound on the bidder dominance
ratio we have that the the number of bid in $S$ that are more than
$p$ is w.h.p much more than $(1-6\gamma)$ times the number of bids
in $T$ that are more than $p$ (and vice versa).

Let $ALG(S,\lceil M/2 \rceil))$ and $ALG(T,\lfloor M/2 \rfloor)$ be
the revenue is generated by the fictitious run of our allocation
algorithm on $S$ and $T$ respectively. Now since the allocation
algorithm is $\alpha$ competitive we have that on expectation
$ALG(S,j) > \alpha OPT(S,j)$. From this it follows that with
probability at least $(1-O(\delta))$ the revenue we earned on
expectation is more than $$ ALG(S,\lceil M/2 \rceil)(1-6\gamma) +
ALG(T,\lfloor M/2 \rfloor)(1-6\gamma) > \alpha(1-6\gamma)(1-
2\gamma)OPT > \alpha(1 - 8\gamma)OPT$$

\section{Conclusion and Open Problems}
The optimal competitive ratio for the allocation problem is open.
\cite{MS} showed an upper bound of $e/(e+1)$ for any randomized
algorithm. The instance for which they show this upper bound is when
there is one bid of 1 and many bids of $\epsilon$. For this
particular instance, the following algorithm gets a competitive
ratio of $2/3$: with probability $1/3$, allocate just one copy and
get a revenue of 1, and with probability $2/3$, run our algorithm.
We conjecture that this algorithm can be generalized to get a $2/3$
competitive ratio. Also, a better  upper bound proof will probably
have to consider instances with multiple peaks, where the ratio of
the $D_i$'s to the $a_{i+1}$'s  is large.

For the auction problem, the
competitive ratio for the unit-demand case is quite small, and that
for the multiple demand case holds only asymptotically. Getting it
to a reasonably large  constant (or proving that it is impossible)
is an important open problem.

The most common scenario in sponsored search auctions is that the
bidders have a constant utility for multiple copies of the item, but
with a daily budget. Our allocation algorithm works for this case as
well, but the reduction from the auction problem is not truthful.
Borgs et al \cite{borgs} give a truthful auction for the offline
case with budgets, using the standard random sampling techniques
with price offers. However, it is not clear how to extend their
auction to the online case. The difficulty is the same as that for
the multiple-demand case, that the price offers are decreasing over
time. But unlike the multiple-demand case, there is no VCG auction
for the budgets case, so our reduction does not work.

\section{Acknowledgements}
The second author would like to thank Jason Hartline for useful
discussions about the digital goods auction, and Deeparnab
Chakrabarty for comments on an earlier draft of the paper.

\bigskip

\part*{Appendix}

\bigskip
\appendix

\section{Proof of Theorem \ref{thm:truthful}}

We will now give the detailed proof of Theorem \ref{thm:truthful}.

Let $Allocate$ be the allocation algorithm that we use as our
subroutine. The algorithm decides to allocate $x(S,j)$ copies at
price $p(S,j)$ when $j$ copies are produced. So the $Allocate(S.j)$
generate revenue $ALG(S,j) = x(S,j)p(S,j)$.

Let the optimal single price allocation algorithm decides to
allocate $x^*(S,j)$ copies at price $p^*(S,j)$ and the optimal
revenue is $OPT(S,i) = x^*(p,i)p^*(S,i)$. Let the $Allocate(S,j)$
have a competitive ratio of $\alpha$. That is for all $S$ and $j$,
$$ALG(S,j) \geq \alpha OPT(S,j)$$ If $M$ is the number of copies of
the item that are finally produced we denote the optimal revenue as
$OPT = OPT({\bf B}, M)$.

\begin{definition}
For any price $p$ let $n(S,p)$, $n(T,p)$ and $n({\bf B}, p)$ be the
set of bids more that $p$ that are made by bidders in $S$, $T$ and
${\bf B}$ respectively.
\end{definition}

Let $Y_i$ be the indicator variable indicating whether the bidder
$i$ is in $S$ or not. Let $f_p(Y_1, \dots, Y_n)$ calculate the
number of bids more than or equal to $p$ that are in $S$, that is
$f_p(Y_1, \dots, Y_n) = n(S,p)$. Note that since the bidders are
randomly placed in $S$ or $T$ we have
$${E}[f_p(Y_1, Y_2, \dots, Y_n)] = \frac{n({\bf B},p)}{2}$$ Let $c_i$
is the maximum change in the value of $f_p$ if we change the value
of $Y_i$. Note that $c_i$ is equal to the number of bids of bidder
$i$ that are more than $p$, that is  $c_i = n(i, p)$. But from our
assumption we have $$ \frac{1}{\eta} < \frac{OPT}{n(i,p)p}$$ So
$n(i, p) < \eta OPT/p$. Hence $$\sum c_i^2 < \frac{\eta OPT}{p}(\sum
c_i) = \frac{\eta OPT}{p}n({\bf B},p)$$

By McDiarmid's Inequality we have for a fixed $p$
\begin{equation} \label{Eq:McDiarmid}
\Pr\left[\left|\frac{n({\bf B},p)}{2}-n(S,p)\right|
> \gamma n({\bf B},p)\right] < \exp\left(\frac{-2\gamma^2n({\bf B},p)^2}{\sum c_i^2}\right)
< \exp\left(\frac{-2p\gamma^2n({\bf B},p)}{\eta
OPT}\right)\end{equation}

\begin{lemma} With probability at least $(1-2|Q|\exp(-2\gamma^2/\eta))$
$$OPT(S,\lceil M/2 \rceil) + OPT(T,\lfloor M/2
\rfloor) > (1 -2\gamma)OPT$$
\end{lemma}

\begin{proof} Note that $n^*({\bf B},p)p > OPT$. From Equation \ref{Eq:McDiarmid}   we have
that for a fixed $p$ if $p = p^*({\bf B},M)$ then
$$\Pr\left[\left|\frac{n({\bf B},p^*({\bf B},M))}{2}-n(S,p^*({\bf B},M))\right|
> \gamma n({\bf B},p^*({\bf B},M))\right] < \exp\left(\frac{-2\gamma^2}{\eta}\right)$$

Now since $p$ takes values from the set $Q$ so by union bound we
have for any $p = p^*({\bf B},M)$ with probability at least
$(1-|Q|exp(-2\gamma^2/\eta))$ we have

$$\left|\frac{n({\bf B},p^*({\bf B},M))}{2}-n(S,p^*({\bf B},M))\right|
> \gamma n({\bf B},p^*({\bf B},M))$$
That is, with probability at least $(1-|Q|\exp(-2\gamma^2/\eta))$
there are more that $(1/2-\gamma)n({\bf B},p^*({\bf B},M))$ bids in
$S$ are more than $p^*({\bf B},M)$, for any $p^*({\bf B},M)$.

Recall that $x^*(S,\rceil M/2 \rceil)$ is the optimal allocation to
bidders in $S$ when $\rceil M/2 \lceil$ . So we have

$$OPT(S, \lceil M/2 \rceil) > \left(\frac{1}{2}-\gamma\right)OPT$$
Similarly with probability at least $(1-|Q|\exp(-2\gamma^2/\eta))$
we have $$OPT(T, \lfloor M/2 \rfloor) >
\left(\frac{1}{2}-\gamma\right)OPT$$
\end{proof}

\begin{corollary} \label{cor:lowerx} With probability at least
$(1-2|Q|\exp(-2\gamma^2/\eta))$ we have the following two
inequalities $$x(S,\lceil M/2 \rceil) > \alpha\left(\frac{1}{2} -
\gamma\right)\frac{OPT}{p(S, \lceil M/2 \rceil}$$

$$x(T,\lfloor M/2 \rfloor) > \alpha\left(\frac{1}{2} -
\gamma\right)\frac{OPT}{p(S, \lfloor M/2 \rfloor}$$
\end{corollary}

From Equation \ref{Eq:McDiarmid} and Corollary \ref{cor:lowerx} we
see that for any fixed $p$ if $p = p(S, \lceil M/2 \rceil)$ then

$$\Pr\left[\left|\frac{n({\bf B},p)}{2}-n(S,p)\right|
> \gamma n({\bf B},p)\right] < \exp\left(\frac{-2\gamma^2p(n({\bf B},p))}{\eta OPT}\right)
< \exp\left( \frac{-2\gamma^2(1/2 - \gamma)}{\eta}\right)$$ Using
union bound we obtain that with probability $( 1 -
2|Q|\exp(-2\gamma^2(1/2 - \gamma)/\eta) -
2|Q|\exp(-2\gamma^2/\eta))$ we have for all $p = p(S, \lceil M/2
\rceil)$
$$\left|\frac{n({\bf B},p)}{2}-n(S,p)\right| > \gamma n({\bf B},p)$$
and for all $p = p(T,\lfloor M/2 \rfloor)$ we have
$$\left|\frac{n({\bf B},p)}{2}-n(T,p)\right| > \gamma n({\bf B},p)$$
For $p = p(S,\rceil M/2 \lceil)$ we have
\begin{equation}\label{eq:nip} n(i, p) < \eta OPT/p < \eta\frac{x(S,
\lceil M/2 \rceil)}{(1/2 -\gamma)\alpha} < \gamma x(S, \lceil M/2
\rceil)\end{equation}

The algorithm decides to allocate $x(S,\lceil M/2 \rceil)$ copies to
bidders in $S$ at price $p(S,\lceil M/2 \rceil)$. So by our
mechanism we allocate $(1-6\gamma)x(S,\lceil M/2 \rceil)$ copies to
bidders in $T$. By the above inequalities we know that with
probability at least $(1 - 4|Q|\exp(-2\gamma^2(1/2-\gamma)/\eta))$
$$n(T, p(S,\lceil M/2 \rceil)) > (1-6\gamma) n(S, p(S,\lceil M/2
\rceil)) + 2\gamma n(S, p(S,\lceil M/2 \rceil))$$ Thus there are at
least $2\gamma n(S, p(S,\lceil M/2 \rceil))$ losing bid in $T$ that
bids more than $p(S,\lceil M/2 \rceil)$. From Equation \ref{eq:nip}
we see that no bidder has more than $\gamma n(S, p(S,\lceil M/2
\rceil))$ bids above $p(S,\lceil M/2 \rceil)$. So by the VCG auction
pricing system each winner in $T$ pays at least $p(S,\lceil M/2
\rceil)$ per copy. So the revenue we get from $T$ is at least
$$p(S,i)(1-6\gamma)x(S,\lceil M/2 \rceil) = ALG(S,\lceil M/2 \rceil)(1-6\gamma)$$
Similarly the revenue we get from $S$ is at least

$$ALG(T,\lfloor M/2 \rfloor)(1-6\gamma)$$
So with probability $(1 - 4|Q|\exp(-2\gamma^2(1/2-\gamma)/\eta))$
our revenue earned is at least

$$ ALG(S,\lceil M/2 \rceil)(1-6\gamma) + ALG(T,\lfloor M/2 \rfloor)(1-6\gamma) > \alpha(1-6\gamma)(1- 2\gamma)OPT > \alpha(1 - 8\gamma)OPT$$
So if $\gamma = \epsilon/8$ and $(2\gamma^2(1/2 - \gamma)/\eta))
> \log(4|Q|/\delta)$ then with probability $(1-\delta)$ the total revenue earned
on expectation is at least $\alpha(1-\epsilon)OPT$.


\begin{thebibliography}{99}

\bibitem{abrams} Zoë Abrams: {\it Revenue maximization when bidders have budgets.} SODA
2006: p.1074-1082

\bibitem{BNDHT} Maria F. Balcan, Nikhil R. Devanur, Jason Hartline,
Kunal Talwar: {\it Random Sampling Auctions for Limited Supply}.
Manuscript, submitted.

\bibitem{blum} Avrim Blum, Jason D. Hartline: {\it Near-optimal online auctions.} SODA
2005: p.1156-1163

\bibitem{borgs} Christian Borgs, Jennifer T. Chayes, Nicole Immorlica, Mohammad
Mahdian, Amin Saberi: {\it Multi-unit auctions with
budget-constrained bidders.} ACM Conference on Electronic Commerce
2005: p.44-51

\bibitem{feige} Uriel Feige, Abraham Flaxman, Jason D. Hartline, Robert D.
Kleinberg: {\it On the Competitive Ratio of the Random Sampling
Auction}. WINE 2005: p.878-886


\bibitem{GH2} Andrew V. Goldberg, Jason D. Hartline: {\it Envy-free auctions for
digital goods.} ACM Conference on Electronic Commerce 2003: p.29-35

\bibitem{GH3} Andrew V. Goldberg, Jason D. Hartline: {\it Competitiveness via
consensus.} SODA 2003: p.215-222

\bibitem{GHKSW} Andrew V. Goldberg, Jason D. Hartline, Anna R. Karlin, Michael E.
Saks, Andrew Wright: {\it Competitive auctions.} Games and Economic
Behavior (55) 2006: p.242-269.



\bibitem{GHW} Andrew V. Goldberg, Jason D. Hartline, Andrew Wright: {\it Competitive
auctions and digital goods.} SODA 2001: p.735-744

\bibitem{HKP} Mohammad Taghi Hajiaghayi, Robert D. Kleinberg, David C. Parkes:
{\it Adaptive limited-supply online auctions.} ACM Conference on
Electronic Commerce 2004: p.71-80

\bibitem{MS} Mohammad Mahdian, Amin Saberi: {\it Multi-unit
Auction with Unknown Supply}. ACM Conference on Electronic Commerce
2006: p.243-249.





\end{thebibliography}
\end{document}